\documentclass[reqno,11pt]{article}

\usepackage{amsmath,amsthm,amssymb,verbatim,color}
\usepackage{graphicx}
\usepackage{wrapfig}
\usepackage{subfigure}
\usepackage{enumerate}
\usepackage{hyperref}

\newtheorem{theorem}{Theorem}
\newtheorem{lemma}[theorem]{Lemma}
\newtheorem{fact}[theorem]{Fact}
\newtheorem{claim}[theorem]{Claim}
\newtheorem{corollary}[theorem]{Corollary}
\newtheorem{conjecture}[theorem]{Conjecture}
\newtheorem{definition}[theorem]{Definition}
\newtheorem{construction}[theorem]{Construction}
\def\beq{\begin{equation}}\def\eeq{\end{equation}}
\def\beqn{\begin{eqnarray}}\def\eeqn{\end{eqnarray}}

\def\({\mbox{$($}}\def\){\mbox{$)$}}

\def\qed{\ifhmode\unskip\nobreak\fi\quad\ifmmode\Box\else$\Box$\fi}
\begin{document}
\title{Perfect matching in $3$-uniform hypergraphs with large vertex degree}

\author{Imdadullah Khan\thanks{Research supported in part by a DIMACS research grant.}\\
\small Department of Computer Science\\[-0.8ex]
\small College of Computing and Information Systems\\[-0.8ex]
\small Umm Al-Qura University\\[-0.8ex]
\small Makkah, Saudi Arabia\\
\small \texttt{iikhan@uqu.edu.sa}\\[-0.8ex]
}

\date{}
\maketitle

\begin{abstract}
A perfect matching in a $3$-uniform hypergraph on $n=3k$ vertices is a subset of $\frac{n}{3}$ disjoint edges. We prove that if $H$ is a $3$-uniform hypergraph on $n=3k$ vertices such that every vertex belongs to at least ${n-1\choose 2} - {2n/3\choose 2}+1$ edges then $H$ contains a perfect matching. We give a construction to show that this result is best possible. 
\end{abstract}

\section{Introduction and Notation}
For graphs we follow the notation in \cite{B1}. For a set $T$, we refer to all of its $k$-element subsets ($k$-sets for short) as ${T \choose k}$ and to the number of such $k$-sets as ${|T| \choose k}$. We say that $H = (V(H),E(H))$ is an $r$-uniform hypergraph or $r$-graph for short, where $V(H)$ is the set of vertices and $E\subset {V(H) \choose r}$, a family of $r$-sets of $V(H)$, is the set of edges of $H$. We say that $H(V_1,\ldots,V_r)$ is an $r$-partite $r$-graph, if there is a partition of $V(H)$ into $r$ sets, i.e. $V(H) = V_1\cup \cdots\cup V_r$ and every edge of $H$ uses exactly one vertex from each $V_i$. We call it a balanced $r$-partite $r$-graph if all $V_i$'s are of the same size. Furthermore $H(V_1,\ldots,V_r)$ is a complete $r$-partite $r$-graph if every $r$-tuple that uses one vertex from each $V_i$ belongs to $E(H)$. We denote a complete balanced $r$-partite $r$-graph by $K^{(r)}(t)$, where $t = |V_i|$. When the graph referred to is clear from the context we will use $V$ instead of $V(H)$ and will identify $H$ with $E(H)$ and $e_r(H) = |E(H)|$. A matching in $H$ is a set of disjoint edges of $H$ and a perfect matching is a matching that contains all vertices. For $U\subset V$, $H|_U$ is the restriction of $H$ to $U$. 

For an $r$-graph $H$ and a set $D = \{v_1,\ldots,v_d\} \in {V \choose d}, 1\leq d \leq r$, the degree of $D$ in $H$, $deg_H(D) = deg_r(D)$ denotes the number of edges of $H$ that contain $D$. For $1\leq d \leq r$, let $$\delta_d =\delta_d(H)= \min\left\{deg_r(D) \; : \; D\in {V \choose d}\right\}.$$

\noindent When $H$ is an $r$-graph and $A$ and $B$ are disjoint subsets of $V(H)$, for a vertex $v\in A$ we denote by $deg_r(v,{B\choose r-1})$ the number of $(r-1)$-sets of $B$ that make edges with $v$, while $d_r(v,{B\choose r-1}) = deg_r(v,{B\choose r-1})/{|B|\choose r-1}$ denotes the density. For such $A$ and $B$, $e_r(A,{B\choose r-1})$ is the sum of $deg_r(v,{B\choose r-1})$ over all $v\in A$ while $d_r(A,{B\choose r-1}) =\dfrac{e_r(A,{B\choose r-1})}{|A|{|B|\choose r-1}}$. We denote by $H(A,{B\choose r-1})$ such an $r$-graph when all edges of $H$ use one vertex from $A$ and $r-1$ vertices from $B$.  When $A_1,\ldots,A_r$ are disjoint subsets of $V$, for a vertex $v\in A_1$ we denote by $deg_r(v, (A_2\times\cdots\times A_r))$ the number of edges in the $r$-partite $r$-graph induced by the subsets $\{v\},A_2,\ldots, A_r$, and $e(A_1, (A_2\times\cdots\times A_r))$ is the sum of $deg_r(v,(A_2\times\cdots\times A_r))$ over all $v \in A_1$. Similarly $$d_r(A_1, (A_2\times\cdots\times A_r)) = \frac{e(A_1, (A_2\times\cdots\times A_r))}{|A_1\times A_2\times\cdots\times A_r|}.$$ 

\noindent An $r$-graph $H$ on $n$ vertices is $\eta$-{\em dense} if it has at least $\eta {n \choose r}$ edges. We use the notation $d_r(H) \geq \eta$ to refer to an $\eta$-dense $r$-graph $H$. A bipartite graph $G=(A,B)$ is $\eta$-{\em dense} if $d(A,B)\geq \eta$.  For $U\subset V$, for simplicity we refer to $d_r(H|_U)$ as $d_r(U)$ and to $E(H|_U)$ as $E(U)$. Throughout the paper $\log$ denotes the base 2 logarithm. Moreover we will only deal with $r$-graphs on $n$ vertices where $n=rk$ for some integer $k$, we denote this by $n\in r\mathbb{Z}$. 

\begin{definition}
 Let $d,r$ and $n$ be integers such that $1\leq d < r$, and $n\in r\mathbb{Z}$. Denote by $m_d(r,n)$ the smallest integer $m$, such that every $r$-graph $H$ on $n$ vertices with $\delta_d(H) \geq m$ contains a perfect matching. 
\end{definition}

For graphs ($r=2$), by the Dirac's theorem on Hamiltonicity of graphs \cite{Dirac1952}, it is easy to see that $m_1(2,n) \leq n/2$, and since the complete bipartite $K_{n/2-1,n/2+1}$ does not have a perfect matching we get $m_1(2,n) = n/2$.  For $r\geq 3$ and $d = r-1$, it follows from a result of R\"{o}dl, Ruci\'{n}ski and Szemer\'{e}di on Hamiltonicity of $r$-graph \cite{RRSz_HAM_ku_colDeg_approx} that $m_{r-1}(r,n) \leq n/2 + o(n)$. K\"{u}hn and Osthus \cite{KO_PM_ku_colDeg} improved this result to $m_{r-1}(r,n) \leq n/2 + 3r^2\sqrt{n\log n}$. This bound was further sharpened in \cite{RRSz_PM_ku_colDeg} to $m_{r-1}(r,n) \leq n/2 + C\log n$. In \cite{RRSz_PM_ku_colDeg_approx_better} the bound was improved to almost the true value; it was proved that $m_{r-1}(r,n) \leq n/2 + r/4$. Finally \cite{RRSz_PM_ku_colDeg_tight} settled the problem for $d=r-1$. K\"{u}hn and Osthus \cite{KO_PM_ku_colDeg} and Aharoni, Georgakopoulos and Spr\"ussel \cite{AGS_partite} studied the minimum degree threshold for perfect matching in $r$-partite $r$-graphs. 

\vskip 10pt
\noindent
The case $d<r-1$ is rather hard. Pikhurko \cite{Pikh_PM_ku_dDeg} proved that for all $d\geq r/2$, $m_d(r,n)$ is close to $\frac{1}{2}{n-d \choose r-d}$. For $1\leq d<r/2$,  H\`{a}n, Person and Schacht \cite{HPS_PM_3u_vertDeg} proved that $$m_d(r,n) \leq \left(\frac{r-d}{r} +o(1)\right){n-d \choose r-d}$$
A recent survey of these and other related results appear in \cite{Rod_Ruc_Survey}. In \cite{HPS_PM_3u_vertDeg} the authors posed the following conjecture. 

\begin{conjecture}[\cite{HPS_PM_3u_vertDeg}, see \cite{Rod_Ruc_Survey} P. 23 ] \label{conject}
For all $1\leq d < r/2$, 
$$m_d(r,n) \sim \max\left\{\frac{1}{2}, 1-\left(\frac{r-1}{r}\right)^{r-d}\right\}{n-d \choose r-d}$$
\end{conjecture}

\noindent Note that for $r=3$ and $d=1$ the above bound yields $$m_1(3,n) \sim \frac{5}{9}{n-1 \choose 2}$$
\noindent Improving an old result of Daykin and H{\"a}ggvist \cite{DH1981}, the authors of \cite{HPS_PM_3u_vertDeg} proved an approximate version of their conjecture for the case $r=3$ and $d=1$, they showed that $m_1(3,n) \leq \left(\frac{5}{9}+o(1)\right){n \choose 2}$ for large $n$. For the case $r=4$ and $d=1$, Markstr\"om and Ruci\'{n}ski \cite{Ruc_Marks_approx4u} proved that $m_1(4,n) \leq  \left(\frac{42}{64}+o(1)\right){n-1 \choose 3}$. Lo and Markstr\"om \cite{Lo_Markst_3partite} determined the exact degree threshold for $r=3$ and $d=1$ for the case of $3$-partite $3$-graphs. 

\vskip6pt
In this paper we settle Conjecture \ref{conject} for the case $r=3$ and $d=1$. Parallel to this work, independently K{\"u}hn, Osthus and Treglown \cite{KOT_parallel} proved the same result. We believe our techniques are more general and have many other applications. In our subsequent work \cite{Khan_4u_vertDeg} we use similar techniques to prove Conjecture \ref{conject} for the case $r=4$ and $d=1$ as well. Our main result in this paper is the following theorem. 
\begin{theorem}\label{ourMainThm}
There exist an integer $n_0$ such that if $H$ is a $3$-graph on $n \geq n_0$ vertices ($n\in 3\mathbb{Z}$), and \beq\label{minDegree}\delta_1(H) \geq {n-1\choose 2} - {2n/3\choose 2} + 1 \eeq  then $H$ has a perfect matching. 
\end{theorem}

\noindent On the other hand the following construction from \cite{HPS_PM_3u_vertDeg} shows that the result is best possible. 

\begin{construction}\label{extConstruction}
Let $H = (V(H),E(H))$ be a $3$-graph on $n$ vertices ($n\in 3\mathbb{Z}$), such that $V(H)$ is partitioned into $A$ and $B$, $|A| = \frac{n}{3}-1$ and $|B|= n-|A|$ and $E(H)$ is the set of all $3$-sets of $V(H)$, $T$, such that $|T\cap A|\geq 1$ (see Figure \ref{ext_example}).
\end{construction}

\noindent We have $\delta_1(H) = {n-1\choose 2} - {2n/3\choose 2}$ (the degree of a vertex in $B$) but since every edge in a matching must use at least one vertex from $A$, the maximum matching in $H$ is of size $|A|=\frac{n}{3}-1$. 

\begin{figure}[h!] 
\centering
\includegraphics[scale=0.50]{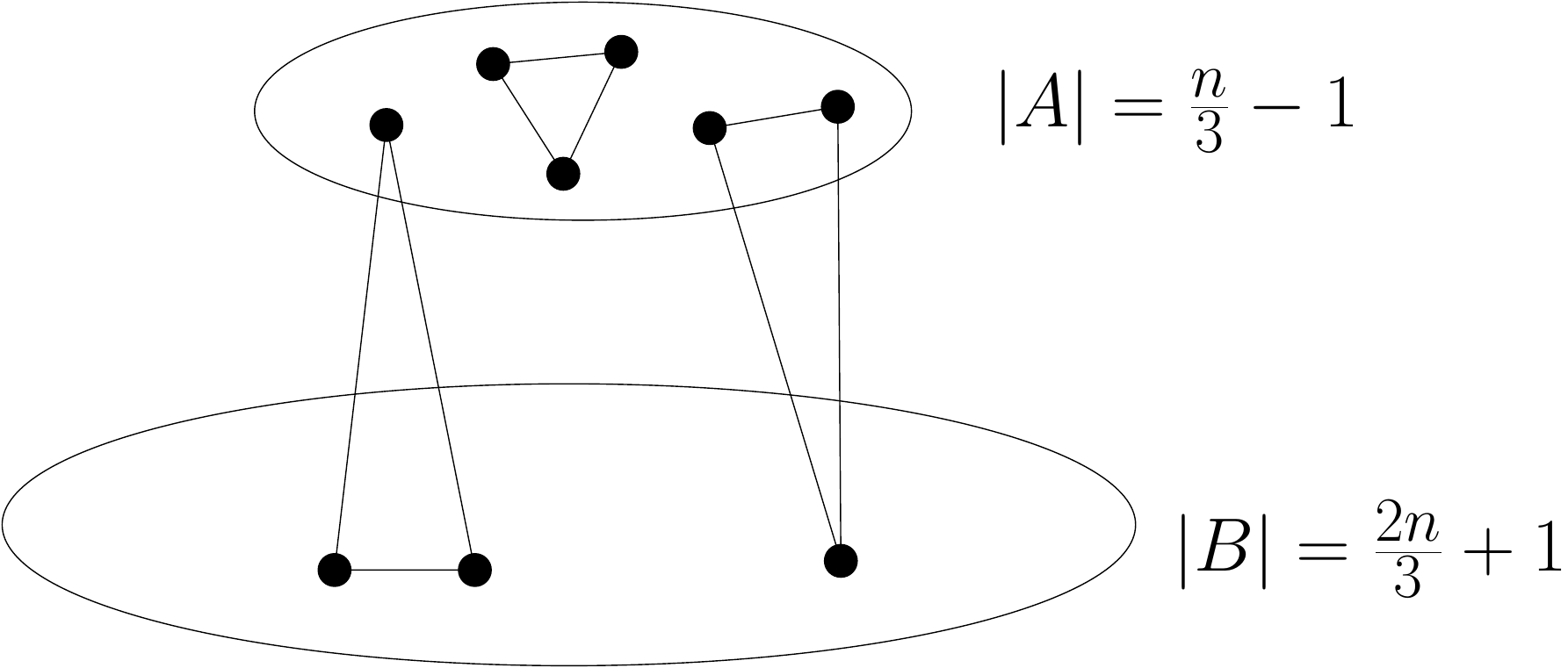} 
\caption{\footnotesize{The extremal example: every edge intersects the set $A$.}}
\label{ext_example}
\end{figure}

\section{The main result}
We distinguish two cases to prove Theorem \ref{ourMainThm}. In Section \ref{non_ext_case} we show that a slightly relaxed minimum degree condition implies  that either $H$ has an {\em `almost perfect matching'} or $H$ is {\em `close to'} the extremal example of Construction \ref{extConstruction}.  In case $H$ is not close to the extremal example we first find an almost perfect matching and extend it to a perfect matching in $H$ using the {\em `absorbing'} technique. On the other hand,  when $H$ is close to the extremal example, in Section \ref{extCase} we build a perfect matching in $H$ with a greedy approach.

\begin{definition}[Extremal Case with parameter $\alpha$]\label{ext_defn} For a constant $0<\alpha < 1$,  we say that $H$ is {\em $\alpha$-extremal}, if the following is satisfied otherwise it is {\em $\alpha$-non-extremal}. There exists a $B\subset V(H)$ such that 
\begin{itemize}
\item $|B|\geq \left(\frac{2}{3}-\alpha\right) n$
\item $d_3\left(B\right) < \alpha$.
\end{itemize}
\end{definition}

\noindent When $H$ is $\alpha$-non-extremal, we use the absorbing lemma, which roughly states that in $H$ there exists a small matching $M$ with the property that every {\em `not too large'}  subset of vertices $W$ can be absorbed into a matching covering $V(M)\cup W$.
\begin{lemma}\label{absorbLemma}(Absorbing Lemma, \cite{HPS_PM_3u_vertDeg}) For every $\eta>0$, there is an integer $n_0 = n_0(\eta)$ such that if $H$ is a $3$-graph on  $n\geq n_0$ vertices with $\delta_1(H)\geq \left(1/2+2\eta\right){n\choose 2}$, then there exists a matching $M$ in $H$ of size $|M|\leq \eta^3n$ such that for every set $W\subset V\setminus V(M)$ of size at most $\eta^6 n\geq |W|\in 3\mathbb{Z}$, there exists a matching covering all the vertices in $V(M)\cup W$.
\end{lemma}

After removing an absorbing matching $M$ from $H$ we find an almost perfect matching in $H|_{V\setminus V(M)}$. The few vertices not covered by this almost perfect matching are absorbed into $M$ to get a perfect matching in $H$. Theorem \ref{optCoverTheorem} in Section \ref{non_ext_case}, using the tools developed in Section \ref{tools}, guarantees the existence of an almost perfect matching in the non-extremal case. 

\begin{theorem}\label{optCoverTheorem}
For all $0<\eta \ll \alpha\ll1$, there is an $n_0$ such that if $H$ is a $3$-graph on $n\geq n_0$ vertices with $$\delta_1(H) \geq \left(\frac{5}{9} - 10\eta\right){n\choose 2},$$ then either 
\begin{itemize}\setlength{\itemsep=-5pt}
\item $H$ contains a matching leaving strictly less than $\eta^2 n$ vertices unmatched or
\item $H$ is $\alpha$-extremal. 
\end{itemize}   
\end{theorem}

When $H$ is $\alpha$-extremal then almost all vertices of $A$ make edges with almost all $3$-sets in ${B\choose 3}$, where $B$ is as in Definition \ref{ext_defn} and $A = V(H)\setminus B$. In Section \ref{extCase} we first match the few vertices of $A$ that do make edges with almost all $3$-sets in ${B\choose 3}$ and the remaining vertices are matched using a K\"onig-Hall type argument. 

\begin{theorem}\label{extCaseTheorem}
For all $0<\alpha\ll1$, there is an $n_0$ such that if $H$ is an $\alpha$-extremal $3$-graph on $n\geq n_0$ vertices with $$ \delta_1(H) \geq {n-1\choose 2} - {2n/3\choose 2}+1,$$ then $H$ contains a perfect matching.
\end{theorem}

\begin{proof}[\textbf{Proof of Theorem \ref{ourMainThm}}]

Let $0<\alpha\ll1$ be given. Applying Lemma \ref{absorbLemma} with parameter $\sqrt{\alpha}$, Theorem \ref{optCoverTheorem} with parameter $\alpha^{3/2}$ and Theorem \ref{extCaseTheorem} with parameter $\alpha$, we get $n_0',n_0''$ and $n_0'''$ respectively. Let $n_0 = 2\max \{n_0',n_0'',n_0'''\}$. Now assume that we have a $3$-graph $H$ on $n\geq n_0$ vertices satisfying (\ref{minDegree}). 
\vskip5pt
From (\ref{minDegree}) when $n$ is large we have $$\delta_1(H)\geq{n-1 \choose 2}-{2n/3\choose 2} +1 > \frac{5}{9}{n-1\choose 2} - \frac{n}{3} > \left(1/2 + 2\sqrt{\alpha}\right){n\choose 2}.$$ Hence $H$ satisfies the conditions of Lemma \ref{absorbLemma} with parameter $\sqrt{\alpha}$. We remove from $H$  an absorbing matching $M$ of size at most $\alpha^{3/2} n$.
\vskip5pt

Let $H' = H|_{V\setminus V(M)}$ be the remaining hypergraph (after removing $M$) on $n' = n - |V(M)|$ vertices. Since ${|V(M)|\choose 2} + |V(M)|\cdot n < 7\alpha^{3/2}{n\choose2} \leq 10\alpha^{3/2}{n'\choose2}$, it is easy to see that $$\delta_1(H')\geq \left(\dfrac{5}{9}-10\alpha^{3/2}\right){n' \choose 2}$$  

As $n' > n_0''$, using Theorem \ref{optCoverTheorem} (with $\eta = \alpha^{3/2}$), in $H'$ we find an almost perfect matching that leaves out a set of at most $\alpha^{3} n' < \alpha^{3} n$  vertices. As guaranteed by Lemma \ref{absorbLemma} the vertices that are left out from this almost perfect matching are absorbed into $M$, and we get a perfect matching in $H$. 

\vskip5pt
In case $H$ is $\alpha$-extremal, using Theorem \ref{extCaseTheorem} we get a perfect matching in $H$, which concludes the proof of Theorem \ref{ourMainThm}. \hfill{}
\end{proof}

\section{Tools}\label{tools}

\noindent We use the following result of Erd\"{o}s \cite{ErdosKST} to find complete balanced $r$-partite subhypergraphs of $r$-graphs.

\begin{lemma}\label{hyperKST}
For every integer $l\geq 1$ there is an integer $n_0 = n_0(r,l)$ such that every $r$-graph on $n> n_0$ vertices, that has at least $n^{r-1/l^{r-1}}$ edges, contains a $K^{(r)}(l)$.
\end{lemma}

\begin{corollary}\label{hyperKST_corr}
\noindent For $0< \eta \ll1$ and $r\leq 3$, if $H$ is an $r$-graph on $n>n_0(\eta,r)$ vertices with $$|E(H)|\geq \eta{n\choose r}$$ then $H$ contains a $K^{(r)}(t)$, where $$t=\eta (\log n)^{1/(r-1)}.$$
\end{corollary}
\noindent This is so because $\eta{n\choose r} \geq \dfrac{\eta n^r}{2r!} \geq \dfrac{n^r}{2^{1/\eta^{r-1}}} = \dfrac{ n^r}{n^{1/(\eta^{r-1} \log n)}} = \dfrac{ n^r}{n^{1/t^{(r-1)}}} $ , as $\eta > \dfrac{2r!}{2^{1/\eta^{(r-1)}}}$.
\vskip6pt

The following lemma is a very useful tool in this section. 

\begin{lemma} \label{subsetsPHP}
Let $m$ be a sufficiently large integer. If $G(A,B)$ is an $\eta$-dense bipartite graph with $|A| = c_1m $ and $B\geq c_2 2^{m}$ for some constants $0<c_1,c_2<1$, then there exists a complete bipartite subgraph $G'(A',B')$ of $G$ such that $A'\subset A$, $B'\subset B$, $|A'|\geq \eta |A|/2 \text{ and } |B'|\geq \dfrac{\eta }{2}\cdot \dfrac{|B|}{2^{c_1m}} \geq \dfrac{\eta c_2 }{2}\cdot 2^{(1-c_1)m} $.
\end{lemma}
\begin{proof}
First we show that there is a set $B_1\subset B$ such that $|B_1| \geq \eta|B|/2$ and for every vertex $ b\in B_1$, $deg(b,A)\geq \eta |A|/2$. Such a subset exists because otherwise the total number of edges in $G$ would be strictly less than $$\frac{\eta|B|}{2}\cdot |A| + |B|\cdot \frac{\eta|A|}{2} = \eta |A||B|$$ a contradiction to the fact that $G(A,B)$ is $\eta$-dense. Now we show that there is the required complete bipartite subgraph in $G(A,B_1)$. To see this consider the neighborhoods in $A$, of the vertices in $B_1$. Since there can be at most $2^{|A|} = 2^{c_1m}$ such
neighborhoods, by averaging there must be a neighborhood  that appears for at least $\dfrac{|B_1|}{2^{c_1 m }}\geq \dfrac{\eta }{2}\cdot \dfrac{|B|}{2^{c_1m}} \geq \dfrac{\eta c_2 }{2} \cdot \dfrac{2^m}{ 2^{c_1 m }}= \dfrac{\eta c_2 }{2}\cdot 2^{(1-c_1)m}$ vertices of $B_1$. Hence we get the desired complete bipartite graph. \hfill{} 
\end{proof}

\vskip6pt\noindent The following two lemmas are repeatedly used in Section \ref{non_ext_case}.
\begin{lemma}\label{3partVolArg}
Let $m$ be a sufficiently large integer and let $H(X,Y,Z)$ be a $3$-partite $3$-graph with $|X| = |Y| = c_1m$ and $|Z| \geq c_2 2^{m^2}$ for some constants $0<c_1,c_2<1$. If $d_3(Z,(X\times Y)) \geq \eta$, then there exists a complete $3$-partite $3$-graph $H'(X',Y',Z')$ as a subgraph of $H$, such that $|X'|=|Y'| =|Z'| \geq \frac{\eta}{4} \log |X|$. 
\end{lemma}
\begin{proof}
First consider the auxiliary bipartite graph $G_1(A,Z)$, where $A =X\times Y$ and a vertex $z\in Z$ is connected to a pair $(a,b)\in A$ if $\{a,b,z\}$ is an edge of $H$. Clearly $G_1$ satisfies the conditions of Lemma\ref{subsetsPHP}.  Applying Lemma \ref{subsetsPHP} on $G_1$ we get a complete bipartite graph $G_2(A',Z')$ such that $A' \subset A= X\times Y$, $Z' \subset Z$, $|A'|\geq \eta |X||Y|/2$ and $|Z'| \geq \frac{c_2\eta}{2} 2^{m^2(1-c_1^2)} > |X|$ where the last inequality follows when $m$ is large.  
\vskip 4pt

Now Let $G_3$ be a graph on vertex set $X\cup Y$ and $(a,b)$ is an edge in $G_3$ if $(a,b)\in A'$. Since $|A'|\geq \eta |X||Y|/2$, we have $|E(G_3)| \geq \dfrac{\eta}{4} {|X\cup Y| \choose 2}$.  Applying Corollary \ref{hyperKST_corr} (for $r=2$), in $G_3$ we get a complete bipartite graph $G_4(X',Y')$ with $X'\subset X$ and $Y'\subset Y$ such that $|X'|=|Y'| \geq \frac{\eta}{4} \log |X|$. Clearly $X'$, $Y'$ and a subset of $Z'$ (of size $|X'|$), correspond to the color classes of required complete $3$-partite $3$-graph. \hfill{} 
\end{proof}

\begin{lemma}\label{subsetPHP_hyper}
Let $m$ be a sufficiently large integer and let $H\left(A,{B\choose 2}\right)$ be a $3$-graph such that $|A|= c_1 m$, $|B|\geq c_2 2^{m^2}$, for some constants $0<c_1,c_2<1$. If $d_3\left(A,{B\choose 2}\right)\geq \eta$, then there exists a complete $3$-partite $3$-graph $H'(A',B',B'')$, with $A'\subset A$, $B'\mbox{ and } B''$ are disjoint subsets of $B$ such that $|A'|=|B'|=|B''| = \eta|A|/2$.  
\end{lemma}
\begin{proof}
First consider the auxiliary bipartite graph $G_1(A,P)$, where $P={B\choose 2}$ and a vertex $a\in A$ is connected to a pair $(b_1,b_2) \in P$ if $(a,b_1,b_2)$ is an edge of $H$. Applying Lemma \ref{subsetsPHP} on $G_1$ we get a complete bipartite graph $(A',P')$ in $G_1$ with $A' \subset A$ and $P' \subset P$ such that $|A'| \geq \eta|A|/2$ and  $$|P'|\geq \frac{\eta}{2} \cdot \frac{|P|}{2^{c_1m}}\geq \frac{\eta}{5}\cdot \frac{|B|^2}{2^{c_1m}} \geq  \frac{\eta}{5} \cdot \frac{|B|^2}{ \left(\frac{|B|}{c_2}\right)^{c_1/m}} = \frac{c_2^2\eta}{5} \left(\frac{|B|}{c_2}\right)^{2-c_1/m} \geq |B|^{2-2/\eta|A|}$$
where the last inequality follows when $m$ is large and $\eta$, $c_1$ and $c_2$ are small constants. 

\vskip 4pt Now construct an auxiliary graph $G_2$ where $V(G_2)=B$ and edges of $G_2$ corresponds to pairs in $P'$. Since $|E(G_2) \geq |B|^{2-2/\eta|A|}$, applying lemma \ref{hyperKST} on $G_2$ (for $r=2$) we get a complete bipartite graph with color classes $B'$ and $B''$ each of size $\eta|A|/2$. Clearly $A'$, $B'$, and $B''$ corresponds to color classes of a complete $3$-partite $3$-graph in $H$ as in the statement of the fact. \hfill{} 
\end{proof}



We also use the following simple facts about graphs. 

\begin{lemma}\label{folk_subgraph}
Any graph on $n$ vertices with $m$ edges has a subgraph of minimum degree $m/n$.
\end{lemma}

\begin{lemma}\label{folk_matching}
Any graph on $n$ vertices has a matching of size $\min\{\delta(G),\lfloor\frac{n}{2}\rfloor\}$.
\end{lemma}


%

\section{Proof of Theorem \ref{optCoverTheorem}}\label{non_ext_case}

Let $H$ be a $3$-uniform hypergraph on $n$ vertices where $n$ is sufficiently large and \beq\label{minDeg_nonext} \delta_1(H) \geq \left(\frac{5}{9} - 10\eta\right){n\choose 2}.\eeq
In $H$ we will find an almost perfect matching (covering at least $(1-\eta^2)n$ vertices). In fact, we will prove a much stronger result. We are going to build a cover ${\cal T} = \{T_{1},T_{2},\ldots$\} where each $T_i$ is a disjoint complete $3$-partite $3$-graph in $H$. These complete $3$-partite $3$-graphs will be balanced and will be of the same size. We refer to them as tripartite graphs. We say that such a cover is optimal if it covers at least $(1-\eta^2)n$ vertices. We will show that either we can find an optimal cover or $H$ is $\alpha$-extremal. It is easy to see that such an optimal cover readily gives us a matching in $H$ that leaves out at most $\eta^2 n$ vertices. 

\begin{proof}[Proof of Theorem \ref{optCoverTheorem}]

We begin with a cover ${\cal T}$ obtained by repeatedly applying Lemma \ref{hyperKST} in the remaining part of $H$ as long as there are at least $\eta^2 n$ vertices left and the condition of Lemma \ref{hyperKST} is satisfied, to get disjoint $K^{(3)}(t)$'s where $t = \eta\sqrt{\log (\eta^2n)}$. Note that by Lemma \ref{hyperKST} we can find larger tripartite graphs in $H$ (at least initially) but since we want all tripartite graphs to be of the same size we find all tripartite graphs of size $3t$. 

Identify by ${\cal T}$ the set of tripartite graphs in the cover and let $V({\cal T})$ be the union of vertices in the tripartite graphs in ${\cal T}$.  We refer to $|V({\cal T})|$ as size of the cover and to a subset of ${\cal T}$ as a subcover in ${\cal T}$. Let ${\cal I} = V(H)\setminus V({\cal T})$ be the set of remaining vertices. If ${\cal T}$ is not an optimal cover, then since we cannot apply Lemma \ref{hyperKST} in  $H|_{\cal I}$ (with parameters $\eta$ to get another $K_3(t)$), we must have that  $|{\cal I}| > \eta^2 n$ and \beq \label{denI} d_3(\cal{I}) < \eta. \eeq 

\noindent Note that from (\ref{minDeg_nonext}) and (\ref{denI}) we immediately get that $|V({\cal T})| > \eta n$. We show that if $H$ is $\alpha$-non-extremal and ${\cal T}$ is not optimal then using the iterative procedure outlined below we can significantly increase the size of our cover (by at least $\eta^4n$ vertices).  After every iteration all tripartite graphs in ${\cal T}$ will be of the same size. Furthermore, all tripartite graphs in ${\cal T}$ will be balanced and if the size of a color class in the tripartite graphs at a given iteration is $t$, then after the iteration it will be $\frac{\eta}{4}\log t$. We take $n$ to be sufficiently large so that till the end of the procedure the size of each color class is large enough for Lemma \ref{3partVolArg} and Lemma \ref{subsetPHP_hyper} to be applicable. 

\vskip6pt \noindent Let $T_i = (V_1^i,V_2^i,V_3^i)$ be a tripartite graph in ${\cal T}$. For $1\leq l,k\leq 3$, we say that $T_i$ is {\em $k$-sided}, if $d_3\left(V_l^i,{{\cal I}\choose 2}\right) \geq 2\eta$, for $k$ color classes $V_l^i$ of $T_i$. We will show that most of the tripartite graphs in ${\cal T}$ are at most $1$-sided or we can significantly increase the size of our cover.


\begin{claim}\label{few_2sided}
If the number of vertices in the at least $2$-sided tripartite graphs in ${\cal T}$ is more than $\eta |V({\cal T})|$, then we can increase the size of ${\cal T}$ by at least $\eta^3 n/8$ vertices, such that all tripartite graphs in the cover are balanced and are of the same size. 
\end{claim}

We repeatedly use Claim \ref{few_2sided} to increase the size of our cover as long as the condition of Claim \ref{few_2sided} is satisfied. Hence in at most $8\eta^{-3}$ iterations we either get an optimal cover or the number of vertices in the at least $2$-sided tripartite graphs is reduced to at most $\eta |V({\cal T})|$. For simplicity we still denote the cover by ${\cal T}$ and ${\cal I}  = V(H)\setminus V({\cal T})$. The size of a color class in each tripartite graph is still denoted by $t$. 

Suppose that ${\cal T}$ is not optimal and we cannot apply Claim \ref{few_2sided}, then the number of $2$-sided vertices in ${\cal T}$ is at most $\eta |V({\cal T})|$. Note that if  $T_i \in {\cal T}$ is at most $1$-sided, then by definition $d\left(T_i, {{\cal I}\choose 2}\right) \leq \left(1/3 + 4\eta\right)$. This together with the bound on the number of $2$-sided tripartite graphs gives us \beq\label{edges_T_to_I} e_3\left(V({\cal T}), {{\cal I}\choose 2}\right) \leq \left(\frac{1}{3}+5\eta\right)|V({\cal T})|{|{\cal I}|\choose 2}. \eeq

For the sum of the degrees of vertices in ${\cal I}$, (\ref{minDeg_nonext}) gives us \beq\label{lower}\sum_{v\in {\cal I}} deg_3(v) \geq |{\cal I}|\left(\frac{5}{9} -10\eta \right){n\choose 2}.\eeq On the other hand considering the number of times each edge is counted, by (\ref{denI}) and  (\ref{edges_T_to_I}) 
\begin{align} 
\sum_{v\in {\cal I}} deg_3(v) &= e_3\left({\cal I},{V({\cal T})\choose 2}\right) + 2\cdot e_3\left(V({\cal T}), {{\cal I}\choose 2}\right) + 3\cdot e_3\left(H|_{\cal I}\right)\notag{}\\
&\leq e_3\left({\cal I},{V({\cal T})\choose 2}\right) + 2\cdot \left(\frac{1}{3}+5\eta\right)|V({\cal T})|{|{\cal I}|\choose 2} + 3\cdot \eta{|{\cal I}|\choose 3}\label{upper}
\end{align}
 Combining (\ref{lower}) and (\ref{upper}) we get 

\begin{align*} 
e_3\left({\cal I},{V({\cal T})\choose 2}\right) &\geq |{\cal I}|\;\left(\frac{5}{9} -10\eta \right){n\choose 2} - 2\cdot \left(\frac{1}{3}+5\eta\right)|V({\cal T})|{|{\cal I}|\choose 2} - 3\cdot \eta{|{\cal I}|\choose 3}\\
&\geq |{\cal I}|\left( \left(\frac{5}{9} -10\eta \right){n\choose 2} - 2\cdot \left(\frac{1}{3}+5\eta\right)|V({\cal T})|\frac{|{\cal I}|}{2} - 3\cdot \eta{|{\cal I}|\choose 2}\right)\\
&\geq |{\cal I}| \left(\frac{5}{9} -38\eta \right){|V({\cal T})|\choose 2}. 
\end{align*}

In (\ref{lower}) and (\ref{upper}) using $e_3\left({\cal I},{V({\cal T})\choose 2}\right) \leq |{\cal I}|\cdot {|V({\cal T})|\choose 2}$ and the fact that $\eta$ is a small constant we get $|V({\cal T})| \geq n/2$.  

For a vertex $v\in V(H)$, consider the edges that $v$ makes with pairs of vertices within a tripartite graph. Since the size of a tripartite graph is $3t \leq 3\eta \sqrt{\log (\eta^2n)}$, the  number of pairs of vertices of any tripartite graph $T_i \in \cal{T}$ is $O(\log n)$. Hence the total number of pairs of vertices within the tripartite graphs in ${\cal T}$ is $O(n\log n) = o{n\choose 2}$. We ignore the at most $o{n\choose 3}$ edges that use more than one vertex from a tripartite graph. By the above observation we still have \beq \label{min_cross_density} d_3\left({\cal I},{V({\cal T})\choose 2}\right)\geq \left(\frac{5}{9}-40\eta\right) \eeq 


\noindent Let $T_i = (V_1^i,V_2^i,V_3^i)$ and $T_j = (V_1^j,V_2^j,V_3^j)$ be two tripartite graphs in ${\cal T}$. We say that $\cal{I}$ is {\em connected} to a pair of color classes $(V_p^i,V_q^j)$, $1\leq p,q \leq 3$, if $d_3({\cal I},(V_p^i \times V_q^j)) \geq 2\eta$. For $(T_i,T_j) \in {{\cal T}\choose 2}$ we define the {\em link graph}, $L_{ij}$, to be a balanced bipartite graph, where the vertex set of each color class of $L_{ij}$ corresponds to the color classes in $T_i$ and $T_j$. A pair of vertices is an edge in $L_{ij}$ iff ${\cal I}$ is connected to the corresponding pair of color classes. 
We will use the following fact from \cite{HPS_PM_3u_vertDeg} for the analysis of the link graph. 
\begin{figure}[h!] 
\centering
\includegraphics[scale=0.80]{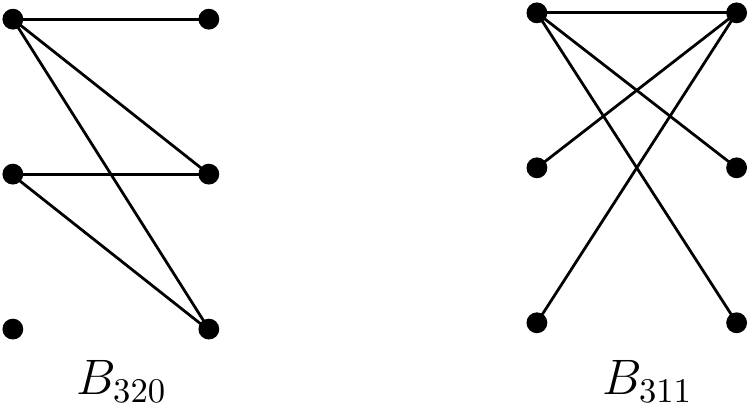} 
\caption{\footnotesize{The balanced bipartite graphs on $6$ vertices, $B_{320}$ and $B_{311}$}}
\label{3x3bipgraphs}
\end{figure}

\begin{fact}\label{classification}
Let $B_{320}$ and $B_{311}$ be as defined in Figure \ref{3x3bipgraphs}. If $B$ is a balanced bipartite graph on $6$ vertices with $|E(B)|\geq 5$, then at least one of the following must be true. 
\begin{itemize}\setlength{\itemsep-2pt}
\item $B$ has a perfect matching.
\item $B$ contains $B_{320}$ as a subgraph.
\item $B$ is isomorphic to $B_{311}$. 
\end{itemize}
\end{fact}

\begin{claim}\label{expanding_byVolArg}
If there are $\eta {|{\cal T}|\choose 2}$ pairs of tripartite graphs $(T_i,T_j)$ such that $L_{ij}$ has a perfect matching or contains a $B_{320}$, then we can increase the size of ${\cal T}$ by at least $\eta^3 n/8$ vertices such that all tripartite graphs in the cover are balanced and are of the same size. \end{claim}

\vskip8pt

We can repeatedly use Claim \ref{expanding_byVolArg} to increase the size of our cover as long as the condition of Claim \ref{expanding_byVolArg} is satisfied. Hence in at most $8\eta^{-3}$ iterations we either get an optimal cover or we have that there are at most $\eta {|{\cal T}|\choose 2}$ pairs of tripartite graphs $(T_i,T_j)$ such that the link graph $L_{ij}$ has at least $5$ edges and contains either a perfect matching or $B_{320}$ as a subgraph.  For simplicity we still denote the cover by ${\cal T}$ and ${\cal I}  = V(H)\setminus V({\cal T})$. 
\vskip8pt

\noindent Assume that ${\cal T}$ is not optimal.  We show that if we cannot apply Claim  \ref{expanding_byVolArg}, then for most of the pairs of tripartite graphs in ${{\cal T}\choose 2}$, the link graph has exactly $5$ edges and is isomorphic to $B_{311}$. Indeed, if there are many ($\geq \sqrt{\eta}$-fraction) pairs of tripartite graphs $(T_i,T_j)$ for which $|E(L_{ij})| \leq 4$ then there is another set of at least $\eta$-fraction of pairs of tripartite graphs $(T_i,T_j)$ for which $|E(L_{ij})| \geq 6$. To see this let ${\cal P}_4 = \{ (T_i,T_j)\in {{\cal T}\choose 2} : |E(L_{ij})| \leq 4\}$ and ${\cal P}_6 = \{ (T_i,T_j)\in {{\cal T}\choose 2} : |E(L_{ij})| \geq 6\}$. Note that for any $(T_i,T_j)\in {\cal P}_4$ by definition $d_3({\cal I}, V(T_i)\times V(T_j)) \leq (4/9 + 10\eta)$. 
\vskip8pt

Now if $|{\cal P}_4| \geq\sqrt{\eta}{|{\cal T}| \choose 2}$ and $|{\cal P}_6| < {\eta}{|{\cal T}| \choose 2}$ then $$d_3\left({\cal I},{V({\cal T})\choose 2}\right)\leq \sqrt{\eta}\cdot \left(\frac{4}{9} + 10\eta\right) + \eta \cdot \frac{9}{9} + (1-\sqrt{\eta})\cdot \left(\frac{5}{9} + 8\eta\right)$$ a contradiction to (\ref{min_cross_density}). Therefore for at least $ (1-\sqrt{\eta}){|{\cal T}|\choose 2}$ pairs of tripartite graphs $(T_i,T_j)\in {{\cal T}\choose 2}$, $L_{ij}$ has at least $5$ edges. Since we cannot apply Claim \ref{expanding_byVolArg}, for at least  $(1-2\sqrt{\eta}) {|{\cal T}|\choose 2}$ of them the $L_{ij}$ is isomorphic to $B_{311}$ and $|E(L_{ij})| =5$.


\begin{claim}\label{noExpand_extremal}
If there are at least $(1-2\sqrt{\eta}) {|{\cal T}|\choose 2}$ pairs of tripartite graphs $(T_i,T_j)$ such that $L_{ij}$ is isomorphic to $B_{311}$, then either 
\begin{itemize}\setlength{\itemsep=-3pt}
\item we can increase the size of ${\cal T}$ by at least $\eta^4 n$ vertices such that all tripartite graphs in the cover are balanced and are of the same size, or 
\item $H$ is $\alpha$-extremal. 
\end{itemize}
\end{claim} We repeatedly use Claim \ref{noExpand_extremal} to increase the size of our cover as long as the condition of Claim \ref{noExpand_extremal} is satisfied. 

Hence proceeding in iterations applying the appropriate claim at each iteration, it is clear that in at most $8\eta^{-4}$ iterations we either get an optimal cover or that $H$ is $\alpha$-extremal. The optimal cover readily gives us an almost perfect matching. \hfill{} \end{proof}

\subsection{Proof of Claim \ref{few_2sided}} Let ${\cal T}'=   \{T_1,T_2,\ldots\} \subset {\cal T} $ be the subcover of the at least $2$-sided tripartite graphs such $|V({\cal T}')|\geq \eta |V({\cal T})| \geq \eta^2 n$. Without loss of generality, say in each $T_i\in {\cal T'}$ we have 
$$d_3\left(V_1^i, {{\cal I}\choose 2}\right)\geq 2\eta \;\; \text{  and  } \;\; d_3\left(V_2^i, {{\cal I}\choose 2}\right)\geq 2\eta.$$ For each such $T_i$, we have  $|V_1^i|=|V_2^i| = t\leq \eta\sqrt{\log (\eta^2n)} = \eta m$ and $|{\cal I}|\geq \eta^2 n > \eta^2 2^{m^2}$. By Lemma \ref{subsetPHP_hyper} (with parameter $\eta$ ) we find two disjoint balanced complete tripartite graphs $(U_1^i,A_1^i,B_1^i)$ and $(U_2^i,A_2^i,B_2^i)$ where $U_1^i$ and $U_2^i$ are subsets of $V_1^i$ and $V_2^i$ respectively, and $A_1^i, A_2^i, B_1^i \mbox{ and } B_2^i$ are disjoint subsets of ${\cal I}$ (see Figure \ref{2sided_extension}). The size of each color class of these tripartite graphs is $\eta |V_1^i|/2$ (we assume it is an integer).  Note that we can find larger tripartite graphs but we keep the size of these new tripartite graphs $3\eta |V_1^i|/2$ only.  We remove the vertices of these new tripartite graphs from their respective sets and add the tripartite graphs to our cover. Removing these vertices from $V_1^i$ and $V_2^i$ creates an imbalance in the leftover part of $T_i$ ($V_3^i$ has more vertices). To restore the balance in the leftover of $T_i$ we discard (add to ${\cal I}$) some arbitrary $|U_1^i| = |U_2^i| = \eta|V_1^i|/2$ vertices from $V_3^i$. The new tripartite graphs use at least $2|A_1^i| + 2|B_1^i| = 2\eta |V_1^i|$ vertices from ${\cal I}$. Therefore,  after discarding the vertices from $V_3^i$ the net increase in the size of our cover is $3\eta |V_1^i|/2$, while all the tripartite graphs in ${\cal T}$ are balanced. 

\begin{figure}[h!] 
\centering
\includegraphics[scale=0.55]{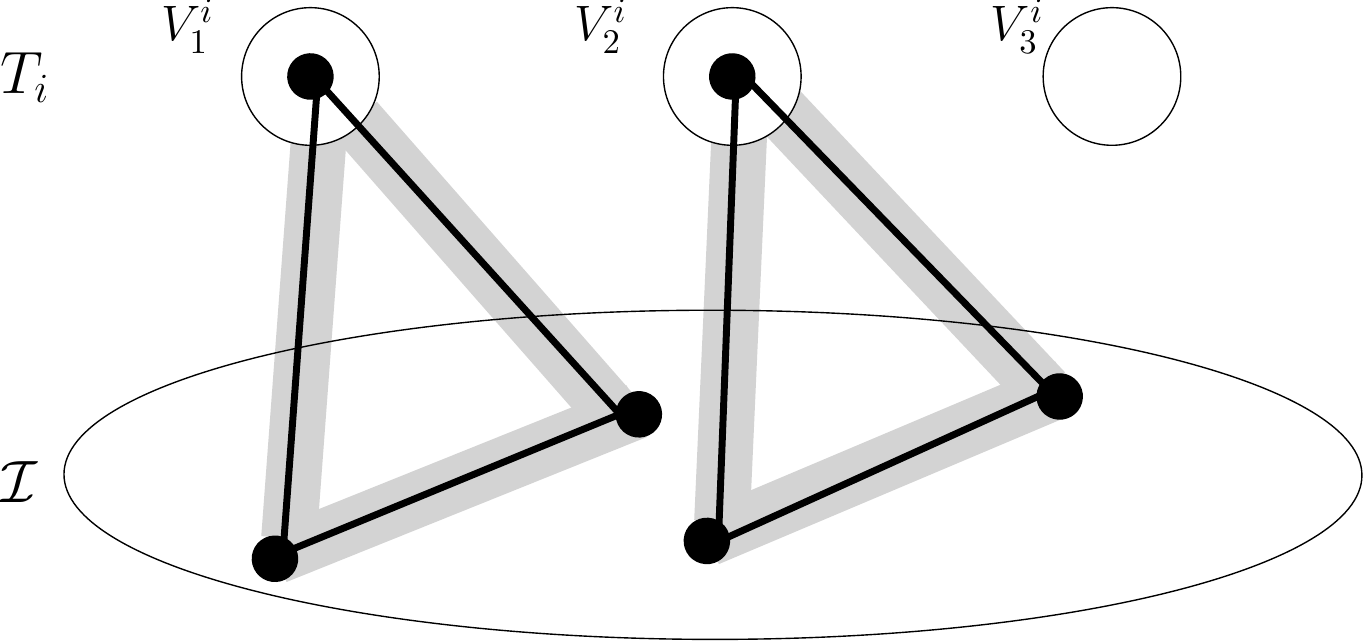} 
\caption{\footnotesize{Shaded triangles represent $2$-sided color classes. Solid  triangles represent the new complete tripartite graphs. $V_3^i$ has extra vertices. }}
\label{2sided_extension}
\end{figure}

We proceed as above for the remaining tripartite graphs in ${\cal T}'$ one by one until we remove at least $\eta |{\cal I}|/8 \geq \eta^3 n/8$ vertices from ${\cal I}$. Since each $T_i \in {\cal T}'$ is at least $2$-sided we can continue as above. Indeed, until we remove $\eta |{\cal I}|/8$ vertices from ${\cal I}$, for the remaining tripartite graphs in ${\cal T}'$ and the remaining part of ${\cal I}$ the condition of Lemma \ref{subsetPHP_hyper} is still satisfied (with parameter $\eta$). Since $|V({\cal T}')|\geq \eta|V({\cal T})|\geq \eta^2 n$, it is easy to see that with this procedure we increase the size of our cover by at least $\eta ^3 n/8$ vertices. Note that the newly made tripartite graphs have color classes of size $\eta t/2$ while the remaining parts of tripartite graphs in ${\cal T}'$ and those in ${\cal T}\setminus {\cal T}'$ are bigger. To make all tripartite graphs in the cover of the same size, we split each tripartite graph in the cover, into disjoint balanced complete tripartite graphs, such that each color class of every tripartite graph is of size $\eta t/2$ (we assume divisibility). \hfill{} \qed

\subsection{Proof of Claim \ref{expanding_byVolArg}}
We first find a set of disjoint pairs of tripartite graphs such that for each pair the link graph either has a perfect matching or contains a $B_{320}$. Consider the auxiliary graph where vertices are the tripartite graphs in ${\cal T}$ and two vertices are connected if for the corresponding tripartite graphs $T_i$ and $T_j$, $L_{ij}$ has a perfect matching or contains a $B_{320}$. Since this auxiliary graph has at least $\eta {|{\cal T}|\choose 2}$ edges, by Lemma \ref{folk_subgraph} and Lemma \ref{folk_matching},  we find a matching of size $\eta |{\cal T}|/3$ in this  graph. Clearly, this matching corresponds to a set of disjoint pairs of tripartite graphs, ${\cal P}\subset {{\cal T}\choose 2}$, such that for the each pair $(T_i,T_j) \in {\cal P}$, $L_{ij}$ has a perfect matching or contains a $B_{320}$. The number of vertices in the $2|{\cal P}|$ tripartite graphs in ${\cal P}$ is at least $\eta n /3$ as $|V({\cal T})| \geq n/2$. We distinguish the following two cases for each pair in ${\cal P}$, to make new tripartite graphs using some vertices from ${\cal I}$.

\vskip 6pt
\textbf{Case 1: {\em $L_{ij}$ has a perfect matching.}} Without loss of generality assume that the perfect matching in $L_{ij}$ corresponds to the pairs $(V_1^i,V_1^j)$, $(V_2^i,V_2^j)$ and $(V_3^i,V_3^j)$. Note that by construction of $L_{ij}$, ${\cal I}$ is {\em connected} to $(V_1^i,V_1^j)$, $(V_2^i,V_2^j)$ and $(V_3^i,V_3^j)$. By definition of connectedness, $d_3({\cal I},V_1^i\times V_2^j)\geq 2\eta$ and $|V_1^i|=|V_1^j| = t\leq \eta\sqrt{\log (\eta^2n)} $ and $|{\cal I}|\geq \eta^2 n > \eta^2 2^{t^2}$, hence the $3$-partite $3$-graph $H(V_1^i,V_1^j,{\cal I})$ satisfies the conditions of Lemma \ref{3partVolArg}.  Applying Lemma \ref{3partVolArg} (with parameter $\eta$) we find a complete balanced tripartite graph $T_1 = ({U_1^i,U_1^j,\cal I}_1)$, such that $$ {\cal I}_1 \subset {\cal I},\;\; {U}_1^i \subset V_1^i \mbox{  ,  } {U}_1^j \subset V_1^j  \;\; \mbox{ and } |{\cal I}_1|=|U_1^i| = |U_1^j| = \frac{\eta}{4}\log t.$$ Similarly, we find such complete balanced tripartite graphs $T_2$ and $T_3$ in $H(V_2^i,V_2^j,{\cal I})$ and $H(V_3^i,V_3^j,{\cal I})$ respectively (see Figure \ref{fig:case1}). Clearly we can have that $T_1, T_2$ and $T_3$ are disjoint from each other as $|{\cal I}|\geq \eta^2 n$ .

\begin{figure}[h!] 
\centering
\includegraphics[page=1,width=3.0in]{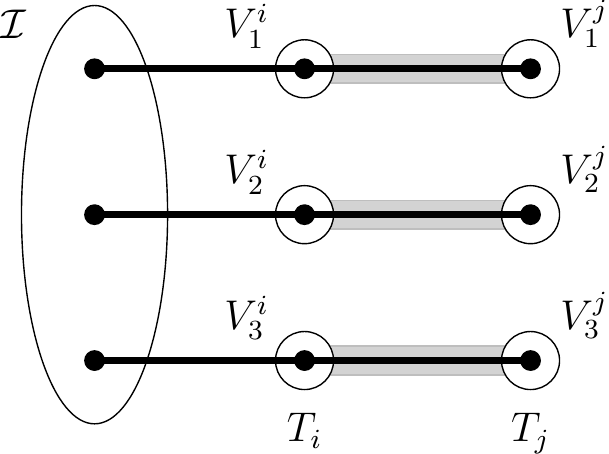}
\caption{\footnotesize{The new tripartite graphs $T_1$, $T_2$ and $T_3$  using vertices from $T_i, T_j$ and ${\cal I}$ when $L_{ij}$ has a perfect matching. Shaded rectangles represent pairs connected to ${\cal I}$. Solid lines represent the new complete tripartite graphs.}}
\label{fig:case1}
\end{figure}

We remove the vertices in $T_1,T_2$ and $T_3$ from their respective sets and add these three new tripartite graphs to our cover. In the remaining parts of $T_i$ and $T_j$ we remove another such set of $3$ disjoint tripartite graphs. By definition of connectedness until we remove at least $\eta^2 t/2$ vertices from each color class of $T_i$ and $T_j$ we still have $d_3({\cal I}, (V_1^i\times V_1^j))> \eta$. Hence by Lemma \ref{3partVolArg} we continue removing such three tripartite graphs until from each color class of $T_i$ and $T_j$ we remove $\eta^2 t/2$ vertices.  Note that the new tripartite graphs use $3\eta^2 t/2$ vertices from ${\cal I}$. Therefore adding these new tripartite graphs to our cover increases the size of the cover by $3\eta t^2/2$ vertices while all tripartite graphs in the cover are still balanced. 

%
%

\vskip 5pt

\textbf{Case 2: {\em $L_{ij}$ contains a $B_{320}$.}} Again without loss of generality assume that in the $B_{320}$ the vertices of degree $3$ and $2$ correspond to the color classes $V_1^i$ and $V_2^i$ respectively. Furthermore, we may assume that ${\cal I}$ is connected to $(V_1^i,V_1^j)$, $(V_1^i,V_2^j)$, $(V_2^i,V_2^j)$ and $(V_2^i,V_3^j)$. By definition of connectedness the $3$-partite subhypergraph of $H$ induced by  ${\cal I}$ and any of the above four pairs of color classes satisfies the conditions of Lemma \ref{3partVolArg}. 

Similarly as in the previous case, applying Lemma \ref{3partVolArg} (with parameter $\eta$) we find the following four disjoint complete tripartite graphs: $(V_{11}^i,V_{11}^j,{\cal I}_1)$, $(V_{22}^i,V_{32}^j,{\cal I}_2)$, $(V_{13}^i,V_{23}^j,{\cal I}_3)$ and $(V_{24}^i,V_{24}^j,{\cal I}_4)$ such that for $1\leq p\leq 3$ and $1\leq q \leq 4$, ${\cal I}_q, V_{pq}^i \mbox{ and }V_{pq}^j$ are disjoint subsets of ${\cal I}, V_{p}^i$ and $V_p^j$ respectively (see Figure \ref{fig:case2}).

For the sizes of these new tripartite graphs we have $$|{\cal I}_1|=|{\cal I}_2|=|V_{11}^i|=|V_{11}^j|=|V_{22}^i|=|V_{32}^j| = \frac{\eta \log t}{4} $$ and $$|{\cal I}_3|=|{\cal I}_4|=|V_{13}^i|=|V_{23}^j|=|V_{24}^i|=|V_{24}^j|= \frac{\eta\log t}{8}.$$ 
\begin{figure}[h!] 
\centering
\includegraphics[page=2,width=3.0in]{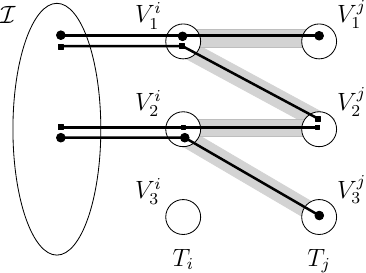}
\caption{\footnotesize{The new tripartite graphs using vertices from $T_i, T_j$ and ${\cal I}$ when $L_{ij}$ contains a $B_{320}$. Shaded rectangles represent pairs connected to ${\cal I}$. Solid lines represent the new complete tripartite graphs. $V_3^i$ has extra vertices.}}
\label{fig:case2}
\end{figure}

In the remaining parts of $T_i$ and $T_j$ we remove another such set of $4$ disjoint tripartite graphs. Again by definition of connectedness and Lemma \ref{3partVolArg} we can continue this process until we remove $\eta^2 t/2$ vertices each from $V_1^i$ and $V_2^i$.  Note that when we remove the vertices of these new tripartite graphs from their respective color classes in $T_i$ and $T_j$, the remaining part of $T_j$ is still balanced while it creates an imbalance in the remaining part of $T_i$, as $V_3^i$ has $\eta^2 t/2$ more vertices than the other two color classes  To restore the balance we discard (add to ${\cal I}$) an arbitrary subset of vertices in $V_3^i$ (of size $\eta^2 t/2$). These new tripartite graphs use at least $\eta^2 t$ vertices from ${\cal I}$. Therefore, after discarding the vertices from $V_3^i$ the net increase in the number of vertices in the cover is $\eta^2 t/2$. 
\vskip 8pt

We proceed in similar manner for all pairs in ${\cal P}$ one by one until we remove at least $\eta |{\cal I}|/8 \geq \eta^3 n/8$ vertices from ${\cal I}$. Applying the appropriate procedure in Case 1 or Case 2 we increase the size of the cover by ${\eta}^{3}n/8$ vertices (as the size of ${\cal P}$ is at least $\eta n/3$ and for every pair the increase is $\eta^2 t/2$)  while keeping all the tripartite graphs in the cover balanced. Note that even after removing $\eta |{\cal I}|/8$ vertices from ${\cal I}$ for the remaining pairs $(T_i,T_j) \in {\cal P}$, the conditions of Lemma  \ref{3partVolArg}   are satisfied (with parameter $\eta$). 

Again as above we make all tripartite graphs in the cover of the same size, by arbitrarily splitting each larger tripartite graph into disjoint tripartite graphs with color classes of size $\eta \log t /4$.  \hfill{} \qed

\subsection{Proof of Claim \ref{noExpand_extremal}}

Similarly as above, first consider the auxiliary graph where vertices are the tripartite graphs in ${\cal T}$ and two vertices are connected if for the corresponding tripartite graphs $T_i$ and $T_j$,  $L_{ij}$ is isomorphic to $B_{311}$. This auxiliary graph has ${\cal T}$ vertices and at least $(1-2\sqrt{\eta}) {|{\cal T}|\choose 2}$ edges. By Lemma \ref{folk_subgraph} and Lemma \ref{folk_matching}, in this graph we can find a matching of size $(1-2\sqrt{\eta}) |{\cal T}|/2$. This matching corresponds to a set of disjoint pairs of tripartite graphs, ${\cal P}_g\subset {{\cal T}\choose 2}$, such that for each pair $(T_i,T_j) \in {\cal P}_g$, $L_{ij}$ is isomorphic to $B_{311}$. Let the set of tripartite graphs in ${\cal P}_g$ be ${\cal T}_g$ and let $V({\cal P}_g)$ be the set of vertices in these tripartite graphs, we have \beq\label{numVertices_good} |V({\cal P}_g)| \geq (1-2\sqrt{\eta})|V({\cal T})|.\eeq

\noindent Without loss of generality assume that for each $(T_i,T_j) \in {\cal P}_g$, the vertices of degree $3$ in the $B_{311}$ in $L_{ij}$ correspond to the color classes $V_1^i$ and $V_1^j$. Thus by definition of connectedness for each $(T_i,T_j) \in {\cal P}_g$, (see Figure \ref{all_B311}) \beq\label{connectedSumm} \text{ for } 1\leq q \leq 3 \;\;\; d_3\left({\cal I},(V_1^i\times V_q^j)\right)\geq 2\eta \;\;\text{ and } d_3\left({\cal I},(V_q^i\times V_1^j)\right)\geq 2\eta.\eeq Let $V_1 = \bigcup\limits_{T_i \in {\cal T}_g} V_1^i$ ($V_2$ and $V_3$ are similarly defined).  We have \beq\label{emptyI_to_V2V3} d_3\left({\cal I},{V_2\cup V_3 \choose 2}\right) < 2\eta. \eeq 
\begin{figure}[h!] 
\centering
\includegraphics[page=4]{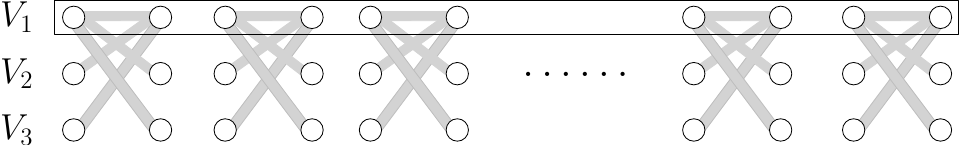}
\caption{\footnotesize{The pairs of tripartite graphs in ${\cal P}_g$.  All $L_{ij}$'s are isomorphic to $B_{311}$. Shaded rectangles represent pairs connected to ${\cal I}$.}}
\label{all_B311}
\end{figure}

For each $(T_i,T_j) \in {\cal P}_g$ the conditions of Lemma \ref{3partVolArg} are satisfied for $H(V_1^i,V_2^j,{\cal I})$, $H(V_1^i,V_3^j,{\cal I})$, $H(V_2^i,V_1^j,{\cal I})$ and $H(V_3^i,V_1^j,{\cal I})$. We find a disjoint complete tripartite graphs in each of these four $3$-partite $3$-graphs (see Figure \ref{B311_ext}). The size of each color class in the new tripartite graphs is $\dfrac{\eta\log t}{4}$. Again we continue removing such sets of four tripartite graphs until we remove $\eta^3 t$ vertices each from $V_1^i$ and $V_1^j$. By construction, these new tripartite graphs remove $\eta^3 t/2$ vertices each from $V_2^i$, $V_3^i$, $V_2^j$  and $V_3^j$.

\begin{figure}[h!] 
\centering
\includegraphics[page=3,width=3.1in]{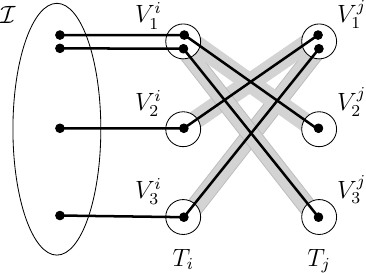}
\caption{\footnotesize{The new tripartite graphs using vertices from $T_i, T_j$ and ${\cal I}$ when $L_{ij}$ is isomorphic to $B_{311}$. Shaded rectangles represent pairs connected to ${\cal I}$. Solid lines represent the new complete tripartite graphs.}}
\label{B311_ext}
\end{figure}
These new tripartite graphs use $2\eta^3t$ vertices from ${\cal I}$. Removing these new tripartite graphs creates an imbalance among the color classes of the remaining parts of $T_i$ and $T_j$, to restore the balance we will have to discard $\eta^3 t/2$ vertices from each color class of $T_i$ and $T_j$ except $V_1^i$ and $V_1^j$. This leaves us with no net gain in the size of the cover. Therefore we will not discard any vertices from these color classes at this time and say that these color classes have $\eta^3 t/2$ extra vertices. 

We proceed in similar manner for each pair in ${\cal P}_g$. Since in total we will use $|{\cal T}_g|\cdot \eta^3 t \leq \eta^3 n \leq \eta |{\cal I}|$ vertices from ${\cal I}$ for the remaining pairs of tripartite graphs in ${\cal P}_g$ the conditions of Lemma  \ref{3partVolArg} (with parameter $\eta$) are satisfied. So we can continue to make new tripartite graphs. 

\vskip 6pt

\noindent Let $V_1^g\subset V_1,V_2^g\subset V_2,V_3^g\subset V_3$ be the union of the corresponding color classes of remaining parts of tripartite graphs in ${\cal T}_g$. Since the number of vertices used in the newly made tripartite graphs above is at most $\eta ^3n$, by (\ref{numVertices_good}) we have \beq\label{sizeV2V3g} |V_2^g| = |V_3^g| \geq (1-3\sqrt{\eta})|V({\cal T})|/3.\eeq Next we show that if at least one of the following density conditions is true then we can increase the size of our cover: \beq\label{emptyV2V3} d_3\left(V_2^g\cup V_3^g\right) \geq \sqrt{\eta},\eeq \beq\label{emptyV2V3_to_I} d_3\left(V_2^g\cup V_3^g,{{\cal I}\choose 2}\right) \geq \sqrt{\eta}.\eeq 

\noindent Assume that $d_3(V_2^g\cup V_3^g) \geq \sqrt{\eta}$. We will show that in $H|_{V_2^g\cup V_3^g}$ there exist disjoint balanced complete tripartite $3$-graphs of size $\eta^3 t/4$ (half the number of extra vertices in a color class) covering at least $\eta^4 n$ vertices. Furthermore, we can find such tripartite graphs in $H|_{V_2^g\cup V_3^g}$ such that from no color class we use more than the number of extra vertices in that color class. 
\vskip 6pt

To see this, call a color class {\em `full'} if these new tripartite graphs use at least $\eta^3 t/4$ vertices. We remove all vertices of each full color class and find tripartite graphs of size $\eta^3 t/4$ in the remaining vertices. Let $k$ be the number of full color classes at a given time and suppose that the total number of vertices covered by the new tripartite graphs is at most $\eta^4 n$. Then, $k \cdot  \eta^3 t/4 < \eta^4 n$ which implies that $tk < 4\eta n$ i.e. the total number of vertices in the full color classes is at most $4\eta n$. Let $H'$ be the remaining part of $H|_{V_2^g\cup V_3^g}$ (after removing all vertices in every full color class). By the above observation $V(H') \geq (1-3\sqrt{\eta})|V({\cal T})|/3 - 4\eta n - \eta^4n$ and $d_3(H') \geq \sqrt{\eta}- 6(\eta +\eta^4) \geq \sqrt{\eta}/2$. Hence by Lemma \ref{hyperKST} we can continue to find complete tripartite graphs of size $\eta^3 t/4$ in $H'$. Note that we do not use more than the number of extra vertices from any color classes. 



\vskip 6pt


We remove some of these new tripartite graphs so that the total number of vertices covered by them is at least $\eta^4n$. Now adding these new tripartite graphs to our cover increases the size of our cover by at least $\eta^4 n$ vertices, as we did not discard vertices from $V_2^g\cup V_3^g$ for rebalancing. Instead the extra vertices are part of these new tripartite graphs. Now in the remaining parts of each $T_i \in {\cal T}_g$ we arbitrarily remove some extra vertices to restore the balance in the tripartite graphs and as above make all tripartite graphs of the same size. \\

\noindent On the other hand if  $d_3(V_2^g\cup V_3^g,{{\cal I}\choose 2}) \geq \sqrt{\eta}$, then since both $|{\cal I}|$ and $|V_2^g\cup V_3^g|$ are at least $\eta^2 n$, by Lemma \ref{hyperKST} we find disjoint complete tripartite graphs with one color class in $V_2^g\cup V_3^g$ and two color classes in ${\cal I}$ covering at least $\eta ^4n$ vertices. Again as above we make these tripartite graphs so as to not use more than the number of extra vertices in any color class. Adding these tripartite graphs increases the size of our cover by at least $\eta^4 n$ vertices.
\vskip 6pt

In case none of the above density conditions hold then by ({\ref{emptyI_to_V2V3}),  (\ref{emptyV2V3}), (\ref{emptyV2V3_to_I}) and the fact that $d_3({\cal I})<\eta$, we get $d_3(V_1^g \cup V_2^g \cup {\cal I}) < 10\sqrt{\eta} < \alpha$. By (\ref{sizeV2V3g}) we have $|V_1^g \cup V_2^g \cup {\cal I}| \geq (2/3 - \alpha)n$. Hence $H$ is $\alpha$-extremal.  This concludes the proof of Claim \ref{noExpand_extremal}. \hfill{} \qed

\section{Proof of Theorem \ref{extCaseTheorem}}\label{extCase}

Let $\alpha$ be given and let $n\in 3\mathbb{Z} \gg \dfrac{1}{\alpha}$.  Our hypergraph $H$ is $\alpha$-extremal i.e. there exists a $B\subset V(H)$ such that 
\begin{itemize}
\item $|B|\geq (\frac{2}{3}-\alpha) n$
\item $d_3(B) < \alpha$.
\end{itemize}
Let $A=V(H)\setminus B$, by shifting some vertices between $A$ and $B$ we can have that $A=n/3$ and $B=2n/3$ (we keep the notation $A$ and $B$). It is easy to see that we still have \beq\label{extDen}d_3(B) < 6\alpha\eeq  
Since we have \beq\label{minDeg_ext} \delta_1(H) \geq {n-1\choose 2} - {2n/3\choose 2}+1={n-1\choose 2} - {|B|\choose 2} + 1\eeq together with (\ref{extDen}) this implies that almost all $3$-sets of $V(H)$ are edges of $H$ except $3$-sets of $B$. Thus roughly speaking almost every vertex $b\in B$ makes edges with almost all pairs of vertices in ${A\choose 2}$ and with almost all pairs of vertices in $B\setminus\{b\} \times A$ and vice versa. Therefore, we will basically match every vertex in $A$ with a distinct pair of vertices in ${B\choose 2}$ to get the perfect matching. However, there may be a few vertices making edges with different pairs of vertices than the typical ones. Hence we will first match those few vertices and then we will use a K{\"o}nig-Hall type argument to match every remaining vertex in $A$ with a distinct pair of remaining vertices in $B$. 

\begin{proof}[Proof of Theorem \ref{extCaseTheorem}]

\vskip4pt

We first identify vertices in $A$ and $B$ that do not satisfy the typical degree conditions as follows.
\begin{definition}
\begin{align*}
\bullet\;\; & X_A \mbox{ {\em (Exceptional vertices in $A$)}} := \{a\in A\; | \;deg_3\left(a,{B\choose 2}\right) < \left(1-\sqrt{\alpha}\right){|B|\choose 2}\}\\
\bullet\;\; & X_B \mbox{ {\em (Exceptional vertices in $B$)}} := \{b\in B \; |\; deg_3(b,(B\times A)) < (1-\sqrt{\alpha})|A|(|B|-1)\}\\
\bullet\;\; & S_A \mbox{ {\em (Strongly Exceptional vertices in $A$)}} := \{a\in A \;|\; deg_3\left(a,{B\choose 2}\right) < {\alpha}^{1/3}{|B|\choose 2}\}\\
\bullet\;\; & S_B \mbox{ {\em (Strongly Exceptional vertices in $B$)}} := \{b\in B \;|\;  deg_3(b,(B\times A)) < {\alpha}^{1/3}|A|(|B|-1)\}
\end{align*}
\end{definition}



\noindent We will show that there are few vertices in $X_A$ and $X_B$ and very few vertices in $S_A$ and $S_B$. 
\begin{claim} We have the following bounds on the sizes of the sets defined above.
\begin{enumerate}[(i)]
\item $|X_A|\leq 18\sqrt{\alpha}|A|$.
\item $|X_B|\leq 18\sqrt{\alpha}|B|$.
\item $|S_B|\leq 40\alpha|B|$.
\item $|S_A|\leq 40\alpha|A|$.
\end{enumerate}
\end{claim}

\proof

We only prove the bounds on $|X_B|$ and $|S_A|$ (the others are similar). Assume that $|X_B|\geq 18\sqrt{\alpha}|B|$. By (\ref{minDeg_ext}) and the definition of $X_B$, for any vertex $b\in X_B$, $deg_3\left(b,{B\choose 2}\right) \geq \sqrt{\alpha}|A|(|B|-1)/2$. Therefore for the number of edges inside $B$ we have $$ 3|E(B)|\geq  |X_B|\cdot\sqrt{\alpha}|A|(|B|-1)/2 \geq 9\sqrt{\alpha}|B|\cdot\sqrt{\alpha}|A|(|B|-1)
\geq  9\alpha |B|(|B|-1)|A|\\
\geq  27\alpha {|B|\choose 3}$$ where the last inequality uses $|A|=|B|/2$. This implies that $d_3(B) \geq 9\alpha$, a contradiction to (\ref{extDen}). 

\vskip 5pt
To see the bound on $|S_A|$, note that by (\ref{minDeg_ext}), if there is a set of $k$ vertices $\{a_1,a_2,\ldots, a_k\}$ in $A$ and a pair $\{b_1,b_2\}$ of vertices in $B$ such that for $1\leq i \leq k$, $\{a_i,b_1,b_2\} \notin E(H)$, then to make up the minimum degree of $b_1$, there are at least $k+1$ edges in $E(B)$ containing $b_1$. Similarly (not necessarily distinct) $k+1$ edges exist in $B$ to make up the minimum degrees of $b_2$. This, together with the fact that every $a\in S_A$ does not make edges with at least $(1-\alpha^{1/3}){|B|\choose 2}$ pairs of vertices in $B$, implies that $$3|E(B)| > |S_A|(1-\alpha^{1/3}{|B|\choose 2}.$$ If $|S_A|>40\alpha|A|$, then $$3|E(B)| > 40\alpha|A|(1-\alpha^{1/3}){|B|\choose 2} = 40\alpha(1-\alpha^{1/3}) \frac{|B|}{2}{|B|\choose 2} \geq 40\alpha{|B|\choose 3}$$ where the last inequality holds when $\alpha$ is a small constant and is a contradiction to (\ref{extDen}). \hfill{}\qed
\vskip4pt


\vskip4pt
\begin{claim}\label{strongExceptional_matching}
There exists a matching $M$ in $H$ such that $M$ covers all the strongly exceptional vertices and if $A' = A\setminus V(M)$, $B'=B\setminus V(M)$ and $n'=n-|V(M)|$, then $|B'| = 2|A'| = 2n'/3$. 
\end{claim}
\begin{proof}
We first show that if both $S_B$ and $S_A$ are non empty, then we can reduce the sizes of both. To see this assume $b\in S_B$ and $a\in S_A$. By definition, $deg_3(a,{B\choose 2})<\alpha^{1/3}{|B| \choose 2}$ and $deg_3(b,(B\times A)) < {\alpha}^{1/3}|A|(|B|-1)$. Hence by (\ref{minDeg_ext}), $deg_3(a,A\times B) \geq (1-2\alpha^{1/3})(|A|-1)|B|$ and $deg_3(b,{B\choose 2}) \geq (1-2\alpha^{1/3}){|B| \choose 2}$.  We can exchange $a$ with $b$ and reduce the size of both $S_B$ and $S_A$, as both $a$ and $b$ are not {\em strongly exceptional} in their new sets. Applying the above procedure we take the sets $A$ and $B$ such that $|S_A|+|S_B|$ is as small as possible (and one of the sets $S_A$ and $S_B$ is empty). 

\vskip4pt First assume that $S_B \neq \emptyset$. As observed above by the minimum degree condition and definition of $S_B$, for every vertex $b \in S_B$,  $deg_3(b,{B\choose 2}) \geq (1-2{\alpha}^{1/3}){|B|\choose 2}$. Since $|S_B|$ is very small and every vertex in $S_B$ makes many edges insides $B$ we can greedily find $|S_B|$ vertex disjoint edges in $H|_B$ each containing exactly one vertex of $S_B$. Indeed after removing at most $|S_B|-1$ disjoint edges from $H|_B$ the remaining vertex in $S_B$ still makes edges with at least $(1-2{\alpha}^{1/3}){|B|\choose 2} - {|S_B|\choose 2} - |B|\times 3|S_B| > 1$ pairs of the remaining vertices. Hence we can greedily match each vertex in $S_B$ in a matching $M$ in $H|_B$. To keep the ratio of the sizes of the remaining parts of $A$ and $B$ intact, we add to $M$, $|S_B|$ other vertex disjoint edges such that each edge has a vertex in $B\setminus X_B$ and the two other vertices are in $A$. We can clearly find such edges because by (\ref{minDeg_ext}) and (\ref{extDen}) almost every vertex in $B\setminus X_B$ makes edges with at least $(1-2\sqrt{\alpha}){|A|\choose 2}$ pairs of vertices in $A$ (as otherwise $d_3(B)$ will be very large). We remove the vertices of $M$ from $A$ and $B$ and by construction $n' = n-6|S_B|$, $|A'|=|A|-2|S_B|$ and $|B'|=|B|-4|S_B|$, hence  $|B'| = 2|A'| = 2n'/3$. 
\vskip8pt

In case $S_A \neq \emptyset$ (and $S_B = \emptyset$), we will find a matching such that each edge contain a vertex in $S_A$. Note that in this case for any vertex $b\in B$ we have $deg_3(b,{B\choose 2}) < \alpha^{1/3}{|B|\choose 2}$. Indeed, if there is a vertex $b\in B$ such that $deg_3(b,{B\choose 2}) \geq \alpha^{1/3}{|B|\choose 2}$ then we can replace $b$ with any vertex $a\in S_A$ to reduce the size of $S_A$ (as the vertex $b$ is not {\em strongly exceptional} in $A$ and $a$ is not {\em strongly exceptional} in the set $B$). We say that vertices in $S_A$ are exchangeable with vertices in $B$ and consider the whole set $S_A\cup B$. By (\ref{minDeg_ext}) for any vertex $v\in S_A\cup B$ 
$$deg_3\left(v,{S_A \cup B\choose 2}\right)\geq {|S_A \cup B|-1\choose 2} -{|B|\choose 2}+1= {|S_A|-1\choose 2} + (|S_A|-1)|B| + 1.$$
 We will prove by induction on $|S_A|$ that we can find a matching $M$ in $H|_{S_A\cup B}$ of size $|S_A|$. Note that this also follows  from a result of Bollob\'as, Daykin and Erd\"{o}s \cite{BDE}.  If $|S_A| = 1$ then clearly we get an edge in $H|_{S_A\cup B}$ and we are done. Now assume that $|S_A|>1$ and that the assertion is true for smaller values of $|S_A|$. Let $v$ be a maximum degree vertex in $H|_{S_A\cup B}$ and let $H'= H|_{S_A\cup B\setminus\{v\}}$.

For any vertex $u\in V(H')$ the number of pairs of vertices in $S_A \cup B$, containing $v$ but not $u$, is at most $|S_A\cup B|-2$. Therefore we get that \begin{align*} \delta_1(H') &\geq {|S_A|-1\choose 2} + (|S_A|-1)|B| + 1 -(|S_A| + |B| -2 )\\ &= {|S_A|-2\choose 2} + (|S_A|-2)|B| + 1\end{align*} where the last equality follows by a simple calculation.
Hence by induction hypothesis there is a matching in  $H'$ of size at least $|S_A| - 1$. Let $M_1$ be a maximum matching in $H'$, if $|M_1|\geq|S_A|$ then we are done so assume that $|M_1| = |S_A|-1$ and every edge in $H'$ intersects $V(M_1)$. This gives us a lower bound on the maximum degree of a vertex in $V(M_1)$ and since $v$ is the overall maximum degree vertex we get

 \begin{align*} deg_3(v) &\geq \frac{|E(H')|}{|V(M_1)|}\geq \frac{|V(H')|\cdot\delta_1(H')}{3|V(M_1)|}\\ &\geq \dfrac{\left(|S_A|+ |B|-1\right)  \left({|S_A|-2\choose 2} + \left(|S_A|-2\right)|B| + 1\right)}{9\left(|S_A|-1\right)} \\ &>  \dfrac{\left(|S_A|+ |B|-1\right)  \left(27\left(|S_A| -1\right)^2\right)}{9\left(|S_A|-1\right)}   \\&>   3\left(|S_A|-1)(|S_A| +|B|-2\right)\end{align*}
where the last inequality uses the fact that $|B|$ is much larger compared to $|S_A|$. Since the last quantity is larger then the number of pairs that use at least one vertex from $V(M_1)$, there is a pair of vertices in  $S_A \cup B\setminus V(M_1)$ that makes an edge with $v$. Adding this edge to $M_1$ we get a matching $M$ which is the required matching. Using the fact that vertices in $S_A$ are exchangeable with vertices in $B$, we get $n' = n-3|S_A|$, $|A'|=|A|-|S_A|$ and $|B'|=|B|-2|S_A|$, hence we get $|B'| = 2|A'| = 2n'/3$. \hfill{} \end{proof}


Having dealt with the \textit{strongly exceptional} vertices, the vertices of $X_A$ and $X_B$ in $A'$ and $B'$ can be eliminated using the fact that their sizes are much smaller than the crossing degrees of vertices in those sets. We have $|X_A|\leq 18\sqrt{\alpha}|A|$ while for any vertex $a\in X_A$, we have that $deg_3(a,{B'\choose 2}) \geq \alpha^{1/3}{|B'|\choose 2}/2$ (because $a\notin S_A$). For each $a\in X_A$ we remove a disjoint edge that contains $a$ and two vertices from $B'$. This can be done greedily as after covering $|X_A|-1$ vertices of $X_A$, the total number of edges removed is at most $50 \sqrt{\alpha} {|B'|\choose 2}$. So there are pairs of vertices in the remaining part of $B'$ making edges with the next vertex of $X_A$.  Similarly for each $b \in X_B$ we remove an edge that contains $b$ and uses one vertex from $A'$ and the other vertex is from $B'$ distinct from $b$. Clearly we can find such disjoint edges by a simple greedy procedure. Hence we removed a partial matching that covers all vertices in the exceptional sets. 

\vskip4pt
\noindent Denote the leftover sets of $A'$ and $B'$ by $A''$ and $B''$ respectively. By construction $|B''|=2|A''|$). We will find $|A''|$ disjoint edges each containing one vertex from $A''$ and two vertices from $B''$. Note that for every vertex  $a\in A''$ we have $deg_3(a,{B''\choose 2})\geq (1-2\alpha^{1/3}){|B''|\choose 2}$ (as $a\notin X_A$). We say that a pair $(b_1,b_2)$ of vertices in $B''$ is {\em good} if $(b_i,b_j,a_k)\in E(H)$ for at least $(1-40{\alpha}^{1/4})|A''|$ vertices $a_k$ in $A''$. Any vertex $b_i \in B''$ makes a good pair with at least $(1-40{\alpha}^{1/4})|B''|$ other vertices in $B''$ (again this is so because $b_i\notin X_B$).

We randomly select a set $P_1$ of $100{\alpha}^{1/4} |B''|$ vertex disjoint good pairs of vertices in $B''$. By the above observation with high probability every vertex $a\in A''$ make edges in $H$ with at least $3|P_1|/4$ pairs in $P_1$ and every pair in $P_1$ makes an edge with at least $3|A''|/4$ vertices in $A''$. In $B''\setminus V(P_1)$ still every vertex makes {good} pairs with almost all other vertices. We pair up each vertex of $B''\setminus V(P_1)$ with a distinct vertex in $B''\setminus V(P_1)$ such that they make a good pair. This can be done by considering a $2$-graph with vertex set $B''\setminus V(P_1)$ and all the good pairs as its edges. A simple application of Dirac\rq{}s theorem on this $2$-graph gives such a perfect matching of vertices in $B''\setminus V(P_1)$. Let the set of these pairs be $P_2$.
  
Now construct an auxiliary bipartite graph $G(L,R)$, such that $L= A''$ and vertices in $R$ correspond to the pairs in $P_1$ and $P_2$. A vertex in $a_k\in L$ is connected to a vertex $y\in R$ if the pair corresponding to $y$ (say $b_i,b_j$) is such that $(b_i,b_j,a_k)\in E(H)$. We will show that $G(L,R)$ satisfies the K{\"o}nig-Hall criteria. Considering the sizes of $A''$ and $P_1$ it is easy to see that for every set $Q\subset R$ if $|Q|\leq (1-40\alpha^{1/4})|A''|$ then $|N(Q)|\geq |Q|$. When $|Q|>(1-40\alpha^{1/4})|A''|$ (using $|B''|=2|A''|$) any such $Q$ must have at least $6|P_1|/10$ vertices corresponding to pairs in $P_1$, hence with high probability $|N(Q)| = |L| \geq |Q|$. Therefore there is a perfect matching of $R$ into $L$. This perfect matching in $G$ readily gives us a matching in $H$ covering all vertices in $A''$ and $B''$, which together with the edges we already removed (covering exceptional  and strongly exceptional vertices) is a perfect matching in $H$. \hfill{} \end{proof}

\end{document}